%% file: Paper.tex
\documentclass[letterpaper,11pt]{article}
\usepackage{amsmath, amssymb}
\usepackage{color}
\usepackage{fullpage}
\usepackage[numbers]{natbib}
\usepackage{enumerate}
\usepackage{hyperref}
\usepackage{graphicx}
\usepackage{caption}
\usepackage{hyperref}
\usepackage{subcaption}
\usepackage{zerosum,csquotes}
\usepackage{enumitem,linegoal}
\usepackage[linesnumbered,ruled,vlined]{algorithm2e}
\usepackage{comment}	
\usepackage{ctable}					
\newtheorem{claim}{Claim}[section]

\newtheorem{observation}{Observation}[section]
\usepackage{mathtools}
\usepackage[flushleft]{threeparttable}
\usepackage[normalem]{ulem}
\usepackage[ruled]{algorithm2e}
\renewcommand{\algorithmcfname}{ALGORITHM}
\usepackage{amsmath,amsfonts,amssymb,bbm} 
\usepackage[noend]{algpseudocode}
\usepackage{enumitem}
\usepackage{comment}

\usepackage{footnote}
\makesavenoteenv{tabular}
\makesavenoteenv{table}

\usepackage[margin=1in]{geometry}

 \newtheorem{theorem}{Theorem}[section]
 \newtheorem{lemma}[theorem]{Lemma}
 
 \newtheorem{corollary}[theorem]{Corollary}
 
 \newtheorem{definition}[theorem]{Definition}

\makeatletter
\def\GrabProofArgument[#1]{ #1: \egroup\ignorespaces}
\def\proof{\noindent\textbf\bgroup Proof%
	\@ifnextchar[{\GrabProofArgument}{. \egroup\ignorespaces}}

\makeatother

\input{macros.tex}

\newcommand*\samethanks[1][\value{footnote}]{\footnotemark[#1]}

\newcounter{proccnt}

\newcommand{\qed}{\hfill $\blacksquare$}
\newcommand{\konote}[1]{}
\definecolor{magenta}{RGB}{0,0,0}
\title{Fair Allocation of Indivisible Items With Externalities}
\author{
	Mohammad Ghodsi \thanks{Sharif University of Technology} \thanks{Institute for Research in Fundamental Sciences (IPM) School of Computer Science }
	\and Hamed Saleh \thanks{University of Maryland}
	\and Masoud Seddighin \samethanks[1]
}

\begin{document}
	\newcommand{\ignore}[1]{}
\renewcommand{\theenumi}{(\roman{enumi}).}
\renewcommand{\labelenumi}{\theenumi}
\sloppy

%
%

\date{}


\maketitle
\begin{abstract}
	\input{abstract}

\end{abstract}
\section{Introduction}
\input{intro}

\section{Model}
\input{model}
\section{Model Evaluation}
\input{justify}

\section{Computing $\mathbf{\EMMS}$}
\label{ecomp}
\input{computational}
\section{$\bf{\alpha}$-$\bf{\EMMS}$ Allocation Problem}

\input{ourwork}


\bibliographystyle{plain}
\bibliography{refs}



\newpage
\appendix
\section{Missing Proofs}
\input{MissingProofs}

\end{document}

%% file: macros.tex
\newcommand{\agents}{\mathcal{N}}
\newcommand{\items}{\mathcal{M}}
\newcommand{\ite}{b}
\newcommand{\agent}{a}

\newcommand{\hug}{\mathcal{H}}
\newcommand{\lon}{\mathcal{L}}

\newcommand{\worst}{\mathcal{W}}
\newcommand{\best}{\mathcal{B}}
\newcommand{\util}{U}

\newcommand{\valu}{V}

\newcommand{\MMS}{\mathsf{MMS}}
\newcommand{\LPT}{\mathsf{LPT}}
\newcommand{\OPT}{O}
\newcommand{\AVG}{\overline{\valu}}
\newcommand{\EMMS}{\mathsf{EMMS}}

%% file: abstract.tex
One of the important yet insufficiently studied subjects in fair allocation is the externality effect among agents. For a resource allocation problem, externalities imply that a bundle allocated to an agent may affect the utilities of other agents.

In this paper, we conduct a study of fair allocation of indivisible goods when the externalities are not negligible. We present a simple and natural model, namely \emph{network externalities}, to capture the externalities. To evaluate fairness in the network externalities model, we generalize the idea behind the notion of maximin-share ($\MMS$) to achieve a new criterion, namely, \emph{extended-maximin-share} ($\EMMS$). Next, we consider two problems concerning our model.

First, we discuss the computational aspects of finding the value of $\EMMS$ for every agent. For this, we introduce a generalized form of partitioning problem that includes many famous partitioning problems such as maximin, minimax, and leximin partitioning problems. We show that a $1/2$-approximation algorithm exists for this partitioning problem.

Next, we investigate on finding approximately optimal $\EMMS$ allocations. That is, allocations that guarantee every agent a utility of at least a fraction of his extended-maximin-share. We show that under a natural assumption that the agents are $\alpha$-self-reliant, an $\alpha/2$-$\EMMS$ allocation always exists. The combination of this with the former result yields a polynomial-time $\alpha/4$-$\EMMS$ allocation algorithm.

%% file: intro.tex
\label{intro}
Consider a scenario where there is a collection of $m$ indivisible goods that are to be divided amongst $n$ agents. For a properly chosen notion of fairness, we desire our division to be fair. Motivating examples are dividing the inherited wealth among heirs, dividing assets of a bankrupt company among creditors, divorce settlements, task assignments, etc.

Fair division has been a central problem in Economic Theory. This subject was first introduced in 1948 by Steinhaus \cite{steinhaus1948problem} in the Polish school of mathematics. The primary model used the metaphor of cake to represent a single divisible resource that must be divided among a set of agents. \emph{Proportionality} is one of the most well-studied notions defined to evaluate the fairness of a cake division protocol. An allocation of a cake to $n$ agents is proportional,
if every agent feels that his allocated share is worth at least $1/n$ of the entire cake.
Despite many positive results regarding proportionality and other fairness notions (e.g. envy-freeness) in cake-cutting (see among many others, \cite{brams1995envy,robertson1998cake, steinhaus1948problem,dehghani2018envy,aziz2016discrete,bei2012optimal}), moving beyond the metaphor of cake the problem becomes more subtle. For example, when the resource is a set of indivisible goods, \textcolor{magenta}{a proportional allocation} is not guaranteed to exist for all instances\footnote{For example, consider the case that there are two agents and the resource is a single indivisible item.}.

For allocation of indivisible goods, Budish \cite{budish2011combinatorial} introduced a new fairness criterion, namely \emph{maximin-share}, that attracted a lot of attention in recent years \cite{amanatidis2015approximation,procaccia2014fair,kurokawa2016can,ghodsi2017fair,suksompong2017approximate}. This notion is a relaxation of proportionality for the case of indivisible items. Assume that we ask agent $i$ to distribute the items into $n$ bundles, and take the bundle with the minimum value. In such a situation, agent $i$ distributes the items in a way that maximizes the value of the minimum bundle. The maximin-share value of agent $i$ is equal to the value of the minimum bundle in the best possible distribution. Formally, the maximin-share of agent $i$, denoted by $\MMS_i$, for a set $\items$ of items and $n$ agents is defined as
$$
\max_{P=\langle P_1,P_2,\ldots,P_n\rangle \in \Pi} \min_j \valu_i(P_j)
,$$
where $\Pi$ is the set of all partitions of $\items$ into $n$ bundles, and $\valu_i(P_j)$ is the value of bundle $P_j$ to agent $i$. \textcolor{magenta}{In a nice paper,} Procaccia and Wang \cite{procaccia2014fair} show that in some instances, no allocation can guarantee maximin-share \textcolor{magenta}{to} all the agents, but an allocation guaranteeing each agent $2/3$ of his maximin-share always exists. This factor \textcolor{magenta}{has} been recently improved to $3/4$ by Ghodsi et al. \cite{ghodsi2017fair}.

Our goal in this paper is to generalize the maximin-share to the case of the agents with externalities. Roughly speaking, externalities are the influences (costs or benefits) incurred by other parties. For resource allocation problems, externalities imply that the bundle allocated to an agent may affect the utility of the other agents. In this work, we assume that the externalities are positive, which is a common assumption in the literature \cite{haghpanah2011optimal,bronzei2013externalities,li2015truthful}.

There are many reasons to consider externalities in an allocation problem. The goods to be divided might exhibit network effects. For example, the value of an XBox to an agent increases as more of his friends also own an XBox, since they can play online. Many merit goods generate positive consumption externalities. In healthcare, individuals who are vaccinated entail positive externalities to other agents around them, since they decrease the risk of contraction. Furthermore, allocating a good to an agent might indirectly affect the utility of his friends since they can borrow it.

\textcolor{magenta}{
The messages of our paper can be condensed as follows:
\textbf{First}, considering the externalities is important: value of $\EMMS$ (the natural generalization we define to adapt $\MMS$ to the environment with externalities) and $\MMS$ might have a large gap. In fact, we show that even a small amount of influence can result in an unbounded gap between these two notions. Thus, when the externalities are not negligible, methods that guarantee $\MMS$ to all the agents might no longer be applicable. 
\textbf{Second}, regarding our model and fairness notion, we can approximately maintain fairness in the environment with externalities. In the next section, we give a more detailed description of our results and the techniques used in the paper. 
}
\subsection{Our Results and Techniques}

In this paper, we take one step toward understanding the impact of externalities in allocation of indivisible items. We start by proposing a general model to capture the externalities in a fair allocation problem under additive assumptions. Although we present some of our results with regard to this general model, the main focus of the paper is on a more restricted model, namely \emph{network externalities}, where the influences imposed by the agents can be represented by a weighted directed graph. \textcolor{magenta}{This model is inspired by the well-studied linear-threshold model in the context of network diffusion.}

We suggest the \emph{extended-maximin-share} notion ($\EMMS$) to adapt maximin-share to \textcolor{magenta}{the} environment with externalities. Similar to maximin-share, our extension is motivated by the maximin strategy in cut-and-choose games. We discuss two aspects of our notion.

First, we discuss the hardness of computing the value of $\EMMS_i$, where $\EMMS_i$ is the extended-maximin-share of agent $i$. For this, we introduce a generalized form of the partition problem that includes many famous partitioning problems such as maximin, minimax, and leximin partitioning problems. This generalized problem is NP-hard due to a trivial reduction from the partition problem. In Section \ref{hardness}, we propose a $1/2$-approximation algorithm for computing $\EMMS_i$ (Theorem \ref{sapx}). In fact, we show that the $\LPT$ method, which is a famous greedy algorithm in the context of job scheduling, guarantees $1/2$-approximation for the general partition problem.  We also reveal several structural properties of such partitions.

Second, we consider the approximate $\alpha$-$\EMMS$ allocation problem. That is, allocations that guarantee every agent a utility of at least a fraction $\alpha$ of his extended-maximin-share.
We define the property of $\beta$-self-reliance and show that when the agents are  $\beta$-self-reliant, there exists an allocation that guarantees every agent $i$ a utility of at least $\beta/2$-$\EMMS_i$ (Theorem \ref{mainres}). This is our most technically involved result.
The basic idea behind our method is as follows: every agent has an expectation value which estimates the utility that he must gain through the algorithm. 
Initially, the expectation value of agent $i$ is at least $\EMMS_i/2$.  In every step of the algorithm, we choose an agent and allocate him a bundle \textcolor{magenta}{with value} at least as his expectation value. Regarding the bundle allocated to this agent, we decrease the expectation value of the remaining agents. Although the algorithm is simple, the analysis is rather complex and heavily exploits the structural properties of the general partitioning problem.
The combination of our existential proof with the $1/2$-approximation algorithm for computing $\EMMS$ yields a polynomial time $\beta/4$-$\EMMS$ allocation algorithm.

\subsubsection*{\textbf{Additional Results}}
 \textcolor{magenta}{Br{\^a}nzei et al.}~\cite{bronzei2013externalities}   extend the proportionality to the case of the agents with externalities. Here, we define the \emph{average-share} notion and claim that average-share is a better extension of proportionality in presence of externalities. We justify our claim by considering the implications among extended-maximin-share, average-share, and extended-proportionality.

\textcolor{magenta}{In interest of space, most of the proofs are deferred to the appendix. }

\subsection{Related Work}
\emph{Maximin-share} has received a lot of attention over the past few years \cite{procaccia2014fair,ghodsi2017fair,amanatidis2015approximation,
gourves2017approximate,kurokawa2016can,farhadi2017fair,suksompong2017approximate,
anari2010equilibrium,branzei2013matchings}. The counter-example suggested by Procaccia and Wang \cite{procaccia2014fair} refutes the existence of any allocation with the maximin-share guarantee. In addition, Procaccia and Wang propose the first approximation algorithm that guarantees each agent $2/3$ of his maximin-share. Recently, Ghodsi et al. \cite{ghodsi2017fair} improve the approximation ratio to $3/4$. For the special case of $3$ agents, Procaccia and Wang \cite{procaccia2014fair} prove that guaranteeing $3/4$ of every agent's maximin-share is always possible. This factor is later improved to $7/8$ by Amanatidis et al. \cite{amanatidis2015approximation} and to $8/9$ by Gourv{\`e}s and Monnot \cite{gourves2017approximate}.
Kurokawa et al. \cite{kurokawa2016can} show that when the valuations are drawn at random, an allocation with maximin-share guarantee exists with a high probability, and it can be found in polynomial time.

Other works generalize maximin-share for different settings. For example, Farhadi et al. \cite{farhadi2017fair} generalize maximin-share for the case of asymmetric agents with different entitlements. They introduce the \emph{weighted-maximin-share} ($\textsf{WMMS}$) criterion and propose an allocation algorithm with a $1/2$-\textsf{WMMS} guarantee.
Suksompong \cite{suksompong2017approximate} considers the case that the items must be allocated to groups of agents. Gourv{\`e}s and Monnot \cite{gourves2017approximate} extend maximin-share to the case that the goods collectively received by the agents satisfy a matroidal constraint and propose an allocation with a $1/2$ maximin-share guarantee.

In recent years, considering externalities for different problems has received an increasing attention in computer science \cite{kempe2008cascade,haghpanah2011optimal,bronzei2013externalities,li2015truthful,
leme2012sequential,anari2010equilibrium,mirrokni2012fixed,branzei2013matchings,farhadi2017fair}. For example, Haghpanah et al. \cite{haghpanah2011optimal} study auction design in the presence of externalities. In a more related work, Br{\^a}nzei et al. \cite{bronzei2013externalities} consider externalities in the cake cutting problem. They introduce a model for cake cutting with externalities and generalize classic fairness criteria to the case with externalities. Following this work, other works also consider externalities in fair division. For example, Li et al. \cite{li2015truthful} study truthful and fair methods for allocating a divisible resource \textcolor{magenta}{with} externalities.

%% file: model.tex
\label{model}
Throughout the paper, we assume $\items $ 
is a set of $m$ indivisible items that must be fairly allocated to a set $\agents = [n]$ of agents, where $[n]$ denotes the set $\{1,2,\ldots,n\}$. We introduce our model in Section \ref{model} and our fairness criteria in Section \ref{fairness}.

\subsection{Modeling the Externalities.}\label{model}
We start by proposing a general model to \textcolor{magenta}{represent} the externalities. In the \textbf{general externalities} model, we suppose that for every set $S$ of items, $\valu_{j,i}(S)$ reflects the utility that agent $i$ \textcolor{magenta}{recieves} by allocating $S$ to agent $j$. In this model, there is no restriction on the value of $\valu_{j,i}(.)$, except that the valuations are additive, i.e.,
$
\valu_{j,i}(S) = \sum_{\ite \in S} \valu_{j,i}(\{\ite\}).
$

We also consider a more restricted model where the externalities are due to the relationships between agents. For example, friends may share their items with a probability which is a function of their relationship. We consider a directed weighted graph $G$ where for every pair of vertices $i$ and $j$,  the weight of edge $(\overrightarrow{j,i})$, denoted by $w_{j,i}$, represents the influence of agent $j$ on agent $i$. We \textcolor{magenta}{refer} such a graph as \emph{influence graph}. If we allocate item $\ite$ to agent $j$, the utility gained by agent $i$ from this allocation would be $\valu_i(\{\ite\}) \cdot w_{j,i} $, where $\valu_i$ is the valuation function of agent $i$. As an example, consider the influence graph illustrated in Figure \ref{graph}. For the allocation that
 allocates $S_i$ to agent $i$ ($1 \leq i \leq 6$), total utility of agent $1$ would be $\valu_1(S_1)\cdot 0.8 + \valu_1(S_2)\cdot 0.2$. We call such a model the \textbf{network externalities} model. Notice that, in this model, the utility of agent $i$ for allocating a set $S$ of items to himself is $\valu_i(S) \cdot w_{i,i}$. In this paper, we suppose \textcolor{magenta}{w.l.o.g.} that the weights of the edges in \textcolor{magenta}{the} influence graph are normalized, so that for every agent $i$, $\sum_j w_{j,i} = 1$. Although we prove some of our results for the general externalities model, our main focus is on the network externalities model.

\begin{definition}
\label{def1}
We say agent $i$ is $\beta$-self-reliant, if $w_{i,i} \geq \beta$.
\end{definition}

For example in Figure \ref{graph}, agent $1$  is $0.8$-self-reliant and agent $5$ is $0.55$-self-reliant.  In real-world situations, we expect $\beta$ to be a value close to $1$.
Note that, being $\beta$-self-reliant for $\beta \simeq 1$ doesn't mean that we can ignore the externalities (for example, see the instance presented in the proof of Observation \ref{gap}).

\begin{figure}[!tbp]
	\centering
	\begin{minipage}[b]{0.49\textwidth}
		\centerline{
			\includegraphics[scale=0.6]{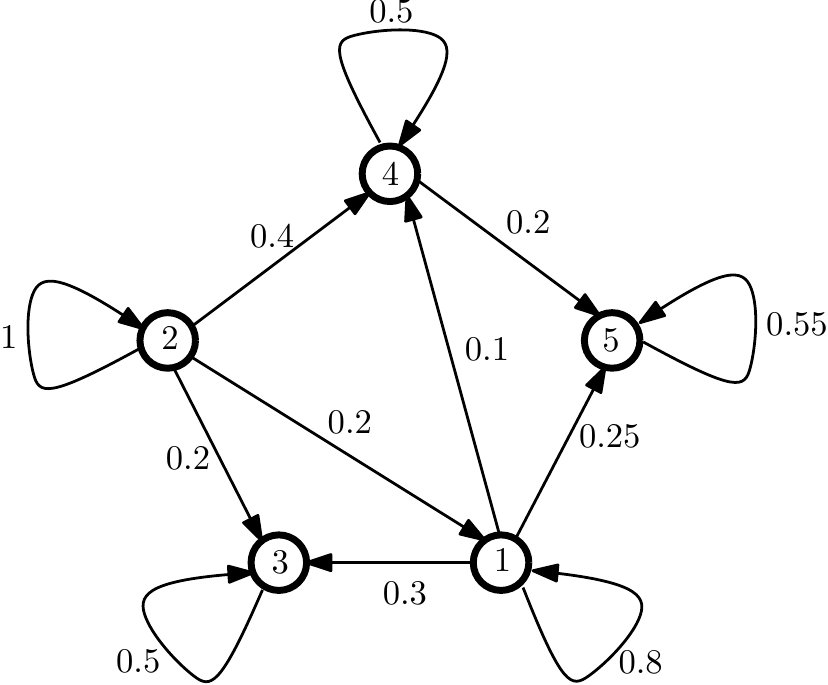}
		}
		\caption{An example of influence graph.}
		\label{graph}
	\end{minipage}
	\begin{minipage}[b]{0.49\textwidth}
		\centerline{
			\includegraphics[scale=0.4]{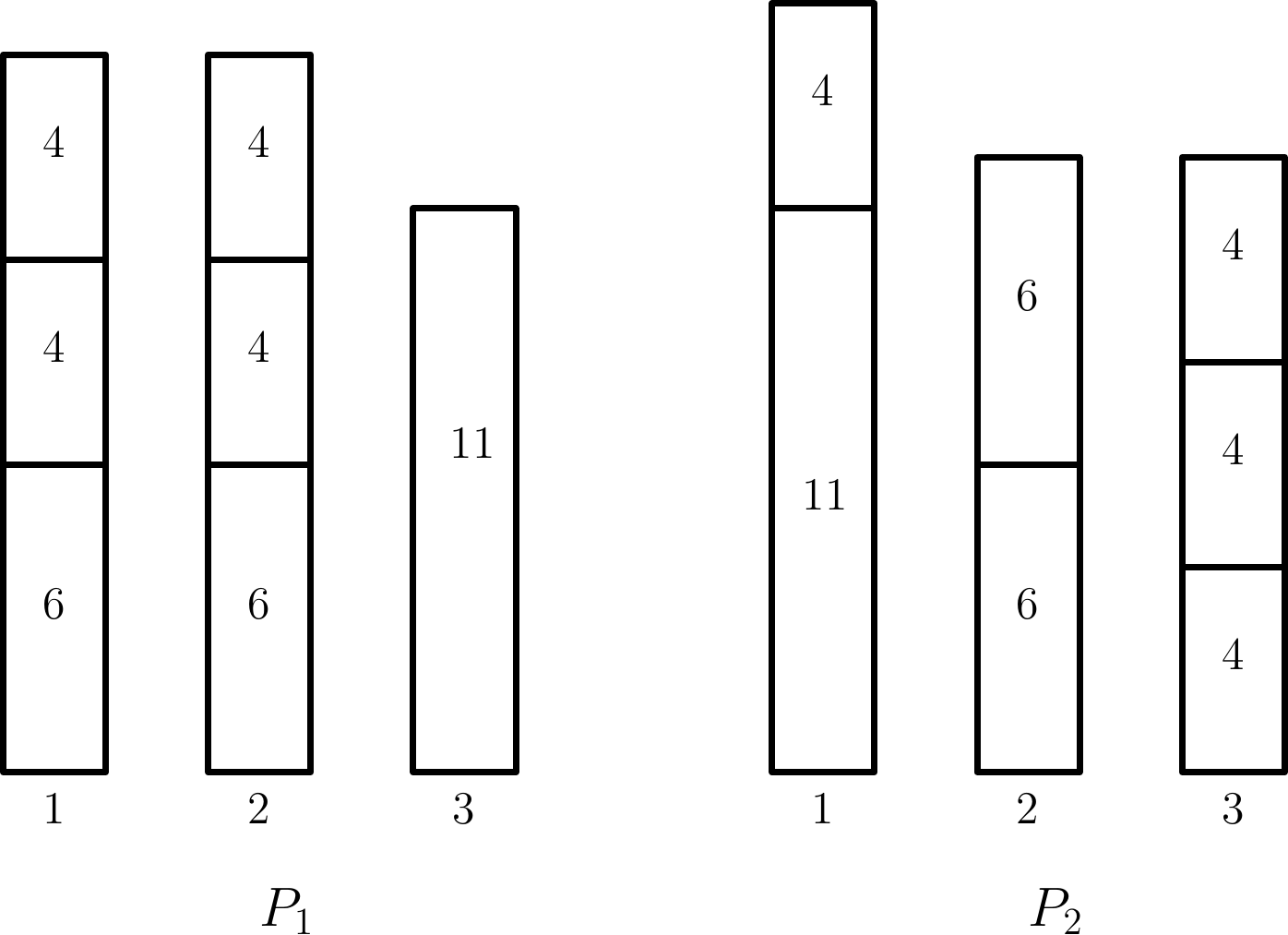}
		}
		\caption{The illustration of two partitions $P_1$ and $P_2$, which are respectively a 
			minimax and a maximin partitions.}
		\label{minimaximin_image}
	\end{minipage}
\end{figure}  

\begin{definition}
For every agent $\agent_i$, we define the \emph{influence vector} of agent $i$, denoted by $x_i = [x_{i,1},x_{i,2},\ldots,x_{i,n} ]$ as the vector representing the influences of the agents on agent $i$ in the influence graph, in non-decreasing order.
\end{definition}

For example, for the graph in Figure \ref{graph}, we have $x_4 = [0,0,0.1,0.4,0.5]$ and $x_5=[0,0,0.2,0.25,0.55]$. Note that when agent $i$ is $\beta$-self-reliant, we have $x_{i,n} \geq \beta$.

\subsection{Fairness Criteria}\label{fairness}
Proportionality and envy-freeness are two of the most common criteria in the literature of cake cutting. For envy-freeness, two extensions are introduced to deal with externalities: \emph{swap-envy-freeness} by Velez \cite{velez2008fairness} and \emph{swap-stability} by Br{\^a}nzei et al. \cite{bronzei2013externalities}. In addition, Br{\^a}nzei et al. \cite{bronzei2013externalities} defined \emph{extended-proportionality} as follows.

\begin{definition}[Extended-proportionality \cite{bronzei2013externalities}]
Let $\hat{V}_i$ be the maximum utility that agent $i$ can gain by allocating each item to the right agent, i.e., agent that maximizes the value of that item for agent $i$. Allocation $A$ is extended-proportional, if the utility of every agent $i$ from $A$ is at least $\hat{V}_i/n$.
\end{definition}

In this paper, we suggest another extension of proportionality, which we call \emph{average-share}.

\begin{definition}[Average-share]
\label{ave}
The average value of item $\ite$ for agent $i$, denoted by $\overline{\valu}_i(\{\ite\})$, is defined as
$
\sum_j \valu_{j,i}(\{\ite\})/n
$. The average-share of agent $i$ is $\overline{\valu}_i(\items)  = \sum_{\ite \in \items} \overline{\valu}_i(\{\ite\}).$ Furthermore, an allocation is said to be average, if the total utility of every agent from this allocation is at least as his average-share. 
\end{definition}

It is easy to observe that both of these notions are equivalent to proportionality in the absence of externalities. However, average-share is a stronger notion, i.e., for every agent $i$,
we have $\overline{\valu}_i(\items) \geq \hat{V}_i/n$.
For the network externalities model, we have  $$\hat{V}_i/n = V_i(\items)\cdot (\max_j w_{j,i})/n , \qquad \mbox{,} \qquad \overline{\valu}_i(\items) = \valu_i(\items)\cdot ( \sum_j w_{j,i})/n.$$

We claim that average-share is a better extension of proportionality to capture the externalities. Note that extended-proportionality suffers from a drawback, that is its low sensitivity to the externalities. For instance, it is reasonable to assume that the best allocation of agent $i$ is to allocate every item to himself. In such a situation, extended-proportionality completely ignores the externalities. We discuss more on this in Section \ref{justify}.

It is worth to mention that both the notions described above are too strong to be \textcolor{magenta}{delivered} in the case of indivisible items. For example, when there are no externalities, no allocation can guarantee \textcolor{magenta}{neither} envy-freeness \textcolor{magenta}{nor} proportionality, or even an approximation of them. Thus, no extension of these notions (including extended-proportionality and average-share) can be guaranteed when items are indivisible.


\subsubsection*{\textbf{Maximin-share}}
In this paper, our \textcolor{magenta}{main} focus is on the maximin-share ($\MMS$) criterion. As mentioned, this notion is introduced by Budish \cite{budish2011combinatorial} as a fairness criteria in division of indivisible items. In Section \ref{intro}, we gave a formal definition of this notion. The intriguing fact about $\MMS$ solution is that it can be motivated by the ``cut and choose'' game. In this game, an agent divides the items into $n$ bundles and lets other agents choose their bundle first. In the worst-case scenario, the least valued bundle remains, and hence the maximin strategy is to divide the items in a way that the minimum bundle is as attractive as possible.  In contrast to proportionality and envy-freeness, guaranteeing a constant fraction of the maximin-share for all the agents is always possible \cite{procaccia2014fair,ghodsi2017fair}.

To extend maximin-share to the case of the agents with externalities, again we consider the worst-case scenario in an ``extended cut and choose'' game. 
Suppose that an agent divides the items into $n$ bundles, and other agents somehow distribute these bundles (one bundle to each agent). The maximin strategy of this agent is to divide the items in a way that maximizes his utility in the worst possible scenario (a scenario that minimizes his utility). In fact, we define the \emph{extended-maximin-share} of each agent $i$ as his outcome in the ``extended cut and choose'' game, regarding maximin strategy.

Formally, let $P = \langle P_1,P_2,\ldots, P_n \rangle $ be a partition of $\items$ into $n$ bundles. Furthermore, let $\mathcal{A}: P \rightarrow [n]$ be an allocation function that allocates every set $P_i$ to agent ${\mathcal{A}(P_i)}$. For brevity, when $P$ is clear from the context, we use $\mathcal{A}_i$ instead of ${\mathcal{A}(P_i)}$ to refer to the agent whom $P_i$ is allocated to. Since exactly one bundle must be allocated to each agent, $\mathcal{A}$ is a bijection. The utility of agent $i$ for an allocation $\mathcal{A}$ is :
$
\util_i({{\mathcal{A}}}) = \sum_j \valu_{{\mathcal{A}_j},{i}}(P_j).
$
The worst allocation of $P$ regarding agent $i$, denoted by $\worst_i(P)$, is the allocation of $P$ that minimizes the utility of agent $i$: 
$
\worst_i(P) = \arg \min_{\mathcal{A} \in \Omega_P} \util_i({\mathcal{A}}),
$
where $\Omega_P$ is the set of all $n!$ different allocations of $P$. Similarly, the best allocation of $P$ is defined as:
$
\best_i(P) = \arg \max_{\mathcal{A} \in \Omega_P} \util_i({\mathcal{A}}).
$
Finally, the extended-maximin-share of agent $i$, denoted by $\EMMS_i$, is defined as:
$$
\EMMS_i = \max_{P \in \Pi} \util_i(\worst_i(P)),
$$
where $\Pi$ is the set of all partitions of $\items$ into $n$ subsets. We also define the \emph{optimal partition} of $\items$ regarding agent $i$, denoted by $\OPT_i$, as the partition that determines the value of $\EMMS_i$,
$
\OPT_i = \arg \max_{P \in \Pi} \util_i(\worst_i(P)).
$
Throughout the paper,  when speaking of the network externalities model, we assume that the bundles in $O_i = \langle O_{i,1}, O_{i,2}, \ldots, O_{i,n} \rangle$ are sorted by their decreasing values for agent $i$, i.e., for all $j$, $\valu_i(O_{i,j}) \geq \valu_i(O_{i,j+1}). $ In addition, when the agent is clear from the context, for any partition $P$ we use $P_j$ to refer to the $j$'th valuable bundle of $P$, regarding that agent.

Finally, an $\alpha$-$\EMMS$ fair allocation problem with the externalities is defined as follows: is there an allocation such that every agent $i$ receives a utility of at least $\alpha\cdot\EMMS_i$?

%% file: justify.tex
\label{justify}
In Section \ref{model}, we introduced three notions: extended-proportionality, average-share, and extended-maximin-share. For a better understanding of these notions, here we briefly compare them in the \textbf{general externalities} model. We already know that average-share is stronger than extended-proportionality. In Lemma \ref{cam}, we prove the same proposition for extended-maximin-share.

\begin{lemma}
	\label{cam}
	Average-share is a stronger notion than extended-maximin-share.
\end{lemma}

By a similar argument \textcolor{magenta}{as in the proof of Lemma \ref{cam}}, we can show that for an arbitrary partition $P$, $\overline{\valu}_i(\items) \leq \util_i(\best_i(P))$. Therefore, for any partition $P$ we have
$\EMMS_i \leq \util_i(\best_i(P)).$
Lemma \ref{cam} states that extended-maximin-share is implied by average-share. However, as we show in Lemma \ref{cmp}, there is no implication between extended-proportionality and extended-maximin-share. 

\begin{lemma}
	\label{cmp}
	Extended-maximin-share does not imply extended-proportionality, nor vice versa. 
\end{lemma}

The fact that in the case without externalities, proportionality is stronger than maximin-share ($\MMS$),  inspires the idea that average-share is a more appropriate extension of proportionality for the case with externalities. In addition, comparing two scenarios in the proof of Lemma \ref{cmp} reveals that the extended-proportionality has a low sensitivity to the externalities.  

In the last part of this section, we show that for $n=2$, a simple \emph{cut and choose} method guarantees $\EMMS_i$ to both the agents. Note that there are instances in which neither extended-proportionality nor average-share can be guaranteed even for two agents. 
\begin{lemma}
	\label{cac}
	For two agents, the following simple \emph{cut and choose} algorithm yields a $1$-$\EMMS$ allocation:
	\begin{itemize}
		\item Ask the first agent to partition the items into his optimal partition $O_1$.
		\item Ask the second agent to allocate $O_1$ (one bundle to each agent). 
	\end{itemize}
\end{lemma}

%% file: computational.tex
\label{hardness}
In this section, we study the problem of computing $\EMMS_i$ and \textcolor{magenta}{$\OPT_i$}. A closer look at the model \textcolor{magenta}{reveals} that the challenges to calculate $\EMMS$ are twofold.  One is to find the worst allocation of a given partition, and the other is to find a partition that maximizes the utility of the worst allocation. In Lemma \ref{fnp} and Observation \ref{snp}, we explore the hardness of these problems for the \textbf{general externalities} model. We then focus on the network externalities model and give a constant approximation algorithm for computing $\EMMS$. 

\begin{lemma}
\label{fnp}
Given a partition $P= \langle P_1,P_2,\ldots,P_n\rangle$ of the items in $\items$, the worst allocation of $P$ regarding agent $i$  can be found in polynomial time. 
\end{lemma}

\begin{observation}
\label{snp}
Since finding the maximin partition of a set of items is $NP$-hard \cite{woeginger1997polynomial}, finding the optimal partition of $m$ items and $n$ agents with externalities is also $NP$-hard.
\end{observation}

Woeginger \cite{woeginger1997polynomial} also showed that finding the maximin partition of a set of items without externalities admits a $\textsf{PTAS}$. However, their method does not directly extend to the case with externalities. To the best of our knowledge, finding an approximately optimal partition for an agent in the general externalities model has not been studied before. 

In the case of \emph{network externalities}, our model is easier to deal with. Since the utility of each agent is a convex combination of his \textcolor{magenta}{valuation}, finding the worst allocation $\worst_i(P)$ is trivial: consider an $n$-step allocation algorithm whose every step allocates the most valuable remaining bundle to a currently unallocated agent with the least effect on agent $i$. Hence,
\begin{equation}
\label{eqm}
\util_i(\worst_i(P)) = \sum_j {x_{i,j} \cdot \valu_i(P_j)}.
\end{equation}
Recall that $x_i$ (the influence vector of agent $i$)  is non-decreasing, and the bundles in $P$ are \textcolor{magenta}{sorted} in non-increasing order of their values for agent $i$. This property of the network externalities model allows us to approximate the value of $\EMMS_i$ with a constant ratio, using a simple greedy approach. On top of that, it is possible to infer relations between $\EMMS$ and some previously defined partitioning schemes.

Apart from the allocation of bundles, partitioning the items is another challenge to overcome. By definition, an optimal partition is a partition that maximizes Equation \eqref{eqm}. Finding an optimal partition for a given vector $x$ is in fact, a generalized form of partitioning problems that includes both maximin and minimax partitions. What happens if we partition the items by one of the famous partitioning schemes such as minimax or maximin? A maximin partition is a partition that maximizes the value of the minimum bundle. It is easy to see that a maximin partition is optimal when $x=[ 1, 0, \ldots, 0 ]$. Likewise, minimax partition is a partition that minimizes the value of the maximum bundle, and it is the optimal partition when $x=[ \frac{1}{n-1}, \frac{1}{n-1}, \ldots, \frac{1}{n-1}, 0] $. Another example is the leximin partition. A leximin partition first maximizes the minimum bundle, and subject to this constraint, maximizes the second least valued bundle, and so on. Real-world applications of leximin allocations are recently studied by Kurokawa, Procaccia and Shah \cite{kurokawa2015leximin}. 
For a small enough $\epsilon$, the optimal partition for vector $x=[ \frac{1-\epsilon}{1-\epsilon^n}, \frac{\epsilon-\epsilon^2}{1-\epsilon^n}, \frac{\epsilon^2-\epsilon^3}{1-\epsilon^n},..., \frac{\epsilon^{n-1}- \epsilon^n}{1-\epsilon^n}]$ is a leximin partition. 
For example, in Figure \ref{minimaximin_image}, if we choose $x$ to be $[ 1, 0, 0 ]$, maximin is the optimal partition, and if we choose $x$ to be $[\frac{1}{2}, \frac{1}{2}, 0 ]$, minimax is optimal.

Since none of these partitioning schemes are always optimal, approximating either of them is not desirable. However, the well-known greedy algorithm $\LPT$ \footnote{Longest processing time} provides a partition $L_i = \langle L_{i,1},L_{i,2},\ldots,L_{i,n} \rangle$ for agent $i$, such that $\util_i(\worst_i(L_i))$ is a constant approximation of $\EMMS_i$.  $\LPT$ is a simple greedy algorithm in the context of job scheduling. This algorithm starts with $n$ empty bundles and iteratively puts the most valuable remaining item into the bundle with the minimum total value. It is proved that the partition provided by $\LPT$ is a constant approximation for both maximin and minimax partitions \cite{graham1969bounds,deuermeyer1982scheduling}.

\begin{theorem}
\label{sapx}
For the network externalities model, we have 
\begin{align}
\util_i(\worst_i(L_i)) &\geq  
\EMMS_i/2. \label{eqy}
\end{align}
\end{theorem}

To prove Theorem \ref{sapx}, we label some of the items as \textit{huge}. Huge items are those whose values are at least $\AVG_i(\items)$. 
Denote the set of huge items for agent $i$ by $\hug_i$. 


\begin{claim}
\label{nohuge}
For an instance with no huge items we have
$
\valu_i(L_{i,n}) \geq \AVG_i(\items)/2.
$
\end{claim}

Since $L_{i,n}$ is the least valued bundle of $L_i$ for agent $i$, $\util_i(\worst_i(L_i)) \geq \valu_i(L_{i,n})$. Furthermore, By lemma \ref{cam}, $\AVG_i(\items) \geq \EMMS_i$.
Hence, when there is no huge item, \textcolor{magenta}{regarding Claim \ref{nohuge},} Inequality \eqref{eqy} holds. Thus, \textcolor{magenta}{to prove Theorem \ref{sapx},} it only suffices to consider the instances with huge items.
Note that, when there are huge items in $\items$, \textcolor{magenta}{$\valu_i(L_{i,n}) \geq \AVG_i(\items)/2$} does not necessarily hold. To cope with such a situation, we need to consider some properties for $\OPT_i$. 

\begin{definition}
\label{nicedef}
We call a partition $P$ \textit{nice} \textcolor{magenta}{for} agent $i$, if no item $\ite$ in some bundle $P_j$ exists, such that $\valu_i(P_j) > \valu_i(\{b\}) > \valu_i(P_n)$ (recall that $\valu_i(P_n) = \min_j{\valu_i(P_j)}$). 
\end{definition}

\begin{claim}
\label{subsetswitch}
For any partition $P$, there exists a nice partition $P'$, such that $\util_i(\worst_i(P)) \leq \util_i(\worst_i(P')).$
\end{claim}

In the rest of this paper, we focus on the optimal partitions which are nice. 
Furthermore, it can be easily observed that $L_i$ is also nice. 

In a nice partition $P$ regarding agent $i$, any bundle $P_j$ containing a huge item $b \in \hug_i$ has no other item. Otherwise, $\valu_i(P_j) > \valu_i(\{b\}) \geq \AVG_i(\items)$. Since $\AVG_i(\items) > \valu_i(P_n)$, this is in contradiction with the niceness of $P$. This fact about nice partitions \textcolor{magenta}{(including $L_i$ and $\OPT_i$)} allows us to deal with huge items.
 We are now ready to prove Theorem \ref{sapx}.

\vspace{0.3cm}
\textbf{Proof of Theorem \ref{sapx}.}
We use induction on the number of agents. For $n = 1$, the statement is trivial. 
For $n>1$, if \textcolor{magenta}{$\valu_i(L_{i,n}) \geq \AVG_i(\items)/2$} holds, we have 
$
\util_i(\worst_i(L_i)) \geq \EMMS_i/2.
$
Thus, when $\items$ contains no huge item, by Claim \ref{nohuge}, $L_i$ is a $1/2$-approximation of $\OPT_i$. 
Therefore, it only remains to consider the case that $n>1$ and $\valu_i(L_{i,n})  < \AVG_i(\items)/2$. \textcolor{magenta}{For this case, we know that $\items$ contains at least one huge item.}

Let $B_{\hug_i}(\OPT_i)$ and $B_{\hug_i}(L_i)$ be the set of the bundles containing huge items in $\OPT_i$ and $L_i$  respectively. We know that $B_{\hug_i}(\OPT_i) = B_{\hug_i}(L_i)$, as the bundles in $B_{\hug_i}(O_i)$ and $B_{\hug_i}(L_i)$ do not contain anything but huge items, and  each huge item is the only item within its bundle \textcolor{magenta}{(recall that both $O_i$ and $L_i$ are nice)}. In addition, $B_{\hug_i}(L_i)$ are the $|\hug_i|$ \textcolor{magenta}{most valuable} bundles in $L_i$. Otherwise, a very similar argument as in the proof of Claim \ref{nohuge} yields $\valu_i(L_{i,n}) \geq \AVG_i(\items)/2$.

  Let $\worst_i'(\OPT_i)$ be the worst possible allocation of $\OPT_i$ with the constraint that allocates $|\hug_i|$ huge items to the $|\hug_i|$ agents with the least influence on agent $i$. By definition, $\util_i(\worst_i'(\OPT_i)) \geq \util_i(\worst_i(\OPT_i))$. Moreover, in both $\worst_i'(\OPT_i)$ and $\worst_i(L_i)$, huge items are allocated to the same set of agents, say $\agents_{\hug_i}$. Now, consider the sub-instance with items $\items \setminus B_{\hug_i}(\OPT_i)$ and agents $\agents \setminus \agents_{\hug_i}$. Note that since $\valu_i(L_{i,n})  < \AVG_i(\items)/2$, the set $\items \setminus B_{\hug_i}(\OPT_i)$ (and hence, $\agents \setminus \agents_{\hug_i}$) is non-empty. 

By the induction hypothesis, for this sub-instance, Inequality \eqref{eqy} holds. Now, adding huge items and their corresponding agents back, increases the utility of agent $i$ by the same \textcolor{magenta}{amount} for both of the allocations. Thus, 
$
\util_i(\worst_i(L_i)) \ge 1/2 \cdot \util_i(\worst_i'(\OPT_i))\ge 1/2 \cdot \util_i(\worst_i(\OPT_i)).
$
\qed
\color{magenta}
\color{black}

%% file: ourwork.tex
In this section, we focus on the allocations that guarantee every agent $i$ an approximation of $\EMMS_i$. We start this section by comparing $\EMMS_i$ and $\MMS_i$. Note that the value of $\MMS$ is determined by the least valued bundle in the maximin partition,  while $\EMMS_i = \util_i(\worst_i(\OPT_i)) \geq \util_i(\worst_i(Q_i))$, where $Q_i$ is a maximin partition of agent $i$. In addition, $\util_i(\worst_i(Q_i))$ is a convex combination of the values of several bundles with a value of at least $\MMS_i$, and hence $\EMMS_i \geq \MMS_i$ always holds. Furthermore, Observation \ref{gap} states that the gap between $\EMMS_i$ and $\MMS_i$ could be unbounded even for the instances with $3$ agents.

\begin{observation}
\label{gap}
For any $c\geq1$, there is an instance with $3$ agents, where $\EMMS_1 > c \cdot \MMS_1$.
\end{observation}
   
It is worth to mention that  the proof of Observation \ref{gap} highlights that even for very few externalities, the gap between $\EMMS$ and $\MMS$ \textcolor{magenta}{might be large}. Thus, the external effects are not negligible even if the impacts of the parties on each other are small.
Our main result is stated in Theorem \ref{mainres}. We show that for the network externalities model when all the agents are $\alpha$-self-reliant, an $\alpha/2$-$\EMMS$ allocation always exists.

\begin{theorem}
\label{mainres}
Let $\mathbb C$ be an instance where for every agent $i$, $w_{i,i} \geq \alpha$. Then, $\mathbb C$ admits an $\alpha/2$-$\EMMS$ allocation. 
\end{theorem}

In the rest of this section, we prove Theorem \ref{mainres} by proposing an $\alpha/2$-$\EMMS$ allocation algorithm for the network externalities model with $\alpha$-self-reliant agents. For brevity, we name our algorithm \emph{Bundle Claiming } algorithm ($\textsf{BC}$).

\subsection{Bundle  Claiming  Algorithm ($\textsf{BC}$)}
In this section, we present the ideas and a general description of \textcolor{magenta}{Bundle Claiming} algorithm. First, let us review the definition of $\EMMS_i$. With abuse of notations, we suppose that $v_i$ is a vector representing the values of the bundles in the optimal partition of agent $i$, i.e., ${v}_{i,j} = \valu_i(O_{i,j})$. Recall that the bundles in $O_i$ are sorted by their decreasing values for agent $i$. Hence, for all $j<n$, we have $v_{i,j}\geq v_{i,j+1}$.  Furthermore, by definition, we have
$
\EMMS_i  	=\sum_j x_{i,j}v_{i,j}. 
$

\begin{observation}
\label{obsv1}
For every $k$, we have
$ \sum_{j\geq k}  x_{i,j}  v_{i,j} \leq v_{i,k}.$
\end{observation}

For example, in an instance with $n=6 $, for $k=4$,  Observation  \ref{obsv1} yields $v_{i,4} \geq v_{i,4} x_{i,4} + v_{i,5} x_{i,5} + v_{i,6}x_{i,6}$. Observation \ref{obsv1} is a direct result of the following two facts:  first, for all $j>k$, we have $v_{i,k}\geq v_{i,j}$ and second, $\sum_{j>k} x_{i,j} \leq 1$.


\begin{definition}
\label{expectation}
For every agent $i$, we define $\ell_i$ to be the expectation level of agent $i$. Agent $i$ with expectation level $\ell_i$, has an expectation value of $v_{i,\ell_i}/2$.
\end{definition}

In the beginning of the algorithm, the expectation level of all the agents are set to $1$. Our algorithm consists of $n$ steps. In each step, we find a bundle $B$ with the minimum number of items that meets the expectation of at least one agent. Bundle $B$ meets the expectation of agent $i$, if $\valu_i(B) \geq v_{i,\ell_i}/2$. We allocate $B$ to one of the agents whose expectation is met (we say this agent is satisfied).  Next, we update the expectation levels of the remaining agents.
The updating process is a fairly complex process which we precisely describe in Section \ref{ana}. Roughly speaking, we update the expectation levels in a way that the following property holds during the algorithm:

\vspace{0.3cm}
\begin{center}
\begin{minipage}{0.85\textwidth}
\textbf{External-satisfaction property:} Let $\cal S$ be the set of currently satisfied agents. For every remaining agent $i$ with expectation level $\ell_i$,  there is a partition of the agents in $\cal S$ into $\ell_i$ subsets, namely $N_{i,1},N_{i,2},\ldots,N_{i,\ell_i-1},N_{i,F}$, such that for all $1 \leq j <\ell_i$, the total set of items allocated to the agents in $N_{i,j}$ is worth at least $v_{i,j}/2$ and at most $v_{i,j}$ to agent $i$, and the total set of items allocated to the agents in $N_{i,F}$ is worth less than $v_{i,\ell_i}/2$ to agent $i$.
\end{minipage}
\end{center}
\vspace{0.3cm}

Note that in the updating process, $\ell_i$ may increase by more than one unit. However, for every remaining agent $i$, $\ell_i \leq n$ must also hold. As we show in Section \ref{ana}, during our algorithm, $\ell_i \leq n$ always holds for every agent $i$.
We use the lower-bounds in this property to show that if the external-satisfaction property holds for agent $i$, total amount of externalities incurred by the satisfied agents is at least $\sum_{k < \ell_i} v_{i,k} x_{i,k}/2.$ 
The fact that $\EMMS_i$ is calculated with regard to the worst allocation of $\OPT_i$ \textcolor{magenta}{is the key to} prove this inequality.

Consider one step of the algorithm and suppose that a set $B$ of items is allocated to agent $i$. Since $B$ met the expectation of agent $i$, $\valu_i(B) \geq v_{i,\ell_i}/2$. Furthermore, the utility that agent $i$ gained through the externalities of the satisfied agents is at least $\sum_{k < \ell_i} v_{i,k} x_{i,k}/2$. Assuming that agent $i$ is $\alpha$-self-reliant, his utility is at least
\begin{align}
&\sum_{k < \ell_i} v_{i,k} x_{i,k}/2 + \alpha  v_{i,\ell_i}/2 \nonumber \\
&\geq \sum_{k < \ell_i} v_{i,k} x_{i,k}/2 + \alpha/2 \sum_{k \geq \ell_i} v_{i,k} x_{i,k} &\mbox{(Observation \ref{obsv1})} \nonumber \\
&\geq \alpha/2 \sum_{k} v_{i,k} x_{i,k} &(\alpha \leq 1 )\nonumber\\
&= \alpha/2\EMMS_i.   \label{more}& 
\end{align}

Inequality \eqref{more} ensures that the items allocated to agent $i$ satisfy him.
Furthermore, we use the upper bounds in the external-satisfaction property to prove that the algorithm satisfies all the agents. To show this, it only suffices to prove that in each step of the algorithm there are enough items to meet the expectation of the remaining agents. Consider agent $i$ which has not satisfied yet. The value of the items allocated to the satisfied agents not in $N_{i,F}$ is at most $\sum_{j < \ell_i} v_{i,j}$. Hence, the total value of the remaining items plus the items allocated to the agents in $N_{i,F}$ is at least $\sum_{ j \geq \ell_i } v_{i,j}$. On the other hand, the value of the items allocated to the agents in $N_{i,F}$ is less than $v_{i,\ell_i}/2$. Thus, the value of the remaining items is at least $v_{i,\ell_i}/2$ which is enough to meet the expectation of agent $i$. As said before, our algorithm maintains the property that $\ell_i \leq n$ for every remaining agent $i$.




\begin{algorithm}[t]
	\ForAll{$\agent_j \in \agents$}{
		$\ell_j \gets 1$ \Comment{Initializing expectation levels} \\
	}
	\While{$\agents \neq \varnothing$}{
		\ForAll{$\agent_j \in \agents$ }{ 
			$\Gamma_j \gets $ Minimum sized subset of $\items$, s.t. $V_j(\Gamma_j) \geq 1/2 \cdot v_{j,\ell_j}$
		}
		Allocate the minimum sized $\Gamma_i$ to agent $i$. \Comment{Allocating} \\ 
		Remove agent $i$ from $\agents$, and $\Gamma_i$ from $\items$. \\
		\ForAll{$\agent_j \in \agents$} {
			Add agent $i$ to $N_{j,F}$. \\
			\While {$\valu_j(N_{j,F}) \geq 1/2 \cdot v_{j,l_j} $}{
				Update $M_j$. \Comment{Maintaining external-satisfaction (see Section \ref{ana})} \\
			}
		}
	}
	\caption{Bundle Claiming algorithm}
	\label{alg1}
\end{algorithm}

\color{magenta}
The details of the bundle claiming algorithm is demonstrated in Algorithm \ref{alg1}, which is a high level abstraction of the full algorithm and captures the overall sketch of the algorithm. In the next section, we show how to maintain the external-satisfaction property in the algorithm.

\color{black}
We end this section by showing that $\textsf{BC}$ can be implemented in polynomial time. The only part of the algorithm whose implementation in polynomial time is not trivial is when we want to find a minimum-sized set $B$ meeting the expectation of at least one agent.

\begin{observation} \label{minbun}
	The minimum-sized set that meets the expectation of at least one remaining agent can be found in polynomial time.
\end{observation}

 It is worth to mention that the operations we apply in order to maintain the external-satisfaction property in the next section, are also trivially polynomial time. 
 
 Finally, using $L_i$ \footnote{Partitioning provided by $\LPT$ algorithm} instead of $\OPT_i$ in $\textsf{BC}$ results in an $\alpha/4$-$\EMMS$ allocation algorithm.

\begin{corollary}
\label{col1}
Let $\mathbb C$ be an instance where for every agent $i$, $w_{i,i} \geq \alpha$. Then, an $\alpha/4$-$\EMMS$ allocation for $\mathbb C$ can be found in polynomial time.
\end{corollary}

\subsection{Maintaining the External-satisfaction Property} \label{ana}
Throughout this section, we suppose that $\mathcal S$ is the set of satisfied agents. Furthermore, for each agent $i$ in $\mathcal S$, we denote the bundle allocated to him by $B_i$. We start by giving a detailed explanation of the updating process.
As mentioned in the previous section, the external-satisfaction condition must hold during the entire algorithm.
To maintain this property in the updating process, for every agent $i$, we define a mapping $M_i$ that represents the partitioning of $\cal S$ for agent $i$ (recall the definition of external-satisfaction).

\begin{definition}
For every agent $i$, we define $M_i: \mathcal S \rightarrow \{\OPT_{i,1},\OPT_{i,2}, \ldots, \OPT_{i,n}\} \cup \{F_i\}$ as a mapping that corresponds each satisfied agent to a bundle in the optimal partition of $O_{i}$ or to $F_i$.  Furthermore, we define $N_{i,j}$ as the set of agents that are mapped to $O_{i,j}$ in $M_i$ and $N_{i,F}$ as the set of agents mapped to $F_i$. During the algorithm, we say mapping $M_i$ is valid, if the following conditions hold:

\vspace{0.2cm}
\begin{enumerate}[label=(\roman*)]
\item $\forall j<\ell_i \qquad \sum_{k \in N_{i,j}} \valu_i(B_k) \geq v_{i,j}/2$
\item $\forall j<\ell_i \qquad \sum_{k \in N_{i,j}} \valu_i(B_k) \leq v_{i,j}$
\item \hspace{1.68cm} $\sum_{k \in N_{i,F}} \valu_i(B_k) < v_{i,\ell_i}/2$
\end{enumerate}
\vspace{0.2cm}
\end{definition}


During the algorithm, mapping $M_i$ must remain valid for every unsatisfied agent $i$.
In the beginning, ${\cal S} = \emptyset$ and for every agent $i$, $\ell_i=1$ and hence, $M_i$ is valid. In every step of the algorithm, we satisfy an agent $i$ by a bundle $B_i$ of items. Next, for every unsatisfied agent $j$, we map agent $i$ to $F_j$ in $M_j$, i.e., we set $M_j(i) =F_j$. In fact, $N_{j,F}$ corresponds to the satisfied agents that are not mapped to any bundle of $O_j$ in $M_j$. We use these agents to update $\ell_j$. Throughout the algorithm, whenever the total value of the items allocated to the agents in $N_{j,F}$ reaches $v_{j,\ell_j}/2$, $M_j$ becomes invalid and hence, we need to update $\ell_j$ and $M_j$ \textcolor{magenta}{to reinstate the validity of $M_j$}. 
To do so, we pick a subset $\delta$ of the agents in $N_{j,F}$ with the minimum size to map them to $O_{j,{\ell_j}}$.
Regarding the validity conditions of $M_j$, total value of the items allocated to the agents in $\delta$ must be at least $v_{j,\ell_j}/2$ and at most $v_{j,\ell_j}$ (we call such subset a \emph{compatible set}).
If a compatible set $\delta$ exists, we map the agents in $\delta$ to $O_{j,{\ell_j}}$ in $M_j$ and increase $\ell_j$ by one. However, there may be some cases that no subset of $N_{j,F}$ is compatible. For such cases, we use the argument in Lemma \ref{tweak}.

\begin{lemma}
\label{tweak}
Suppose that total value of the items allocated to the agents in $N_{j,F}$ is at least $v_{j,\ell_j}/2$, but $N_{j,F}$ admits no compatible subset. Then,
it is possible to modify $M_j$ such that conditions (i) and (ii) remain valid for $M_j$ and $N_{j,F}$ contains at least one compatible subset.
\end{lemma}

Note that, after increasing $\ell_j$ for agent $j$, condition (iii) \textcolor{magenta}{may still be violated}. 
 In that case, as long as condition (iii) is violated, we continue updating. Each time we update $M_j$, value of $\ell_j$ is increased by one. Since at least one agent is mapped to $O_{j,\ell}$ for each $\ell < \ell_j$, $\ell_j$ never exceeds $n$. \textcolor{magenta}{In the appendix, you can find a pseudo-code for the updating process (Algorithm \ref{alg2})}



\input{analysis}

%% file: analysis.tex
In the last part of this section, we prove Lemma \ref{satisfaction} which shows that the \textcolor{magenta}{value of the} externalities imposed to agent $i$ by the satisfied agents is lower-bounded by $\sum_{j<\ell_i} x_{i,j}v_{i,j}/2$. As said before, the fact that $\EMMS_i$ is defined with regard to the worst allocation of $O_i$ plays a key role in proving Lemma \ref{satisfaction}.

\begin{lemma}
\label{satisfaction}
Consider one step of the algorithm, and let agent $i$ be an arbitrary remaining agent with $\ell_i > 1$. Then, we have $\sum_{j \in \mathcal{S}} w_{j,i} \cdot \valu_i(B_j) \geq \sum_{j < \ell_i} x_{i,j} \cdot v_{i,j} /2.$
\end{lemma}

%% file: MissingProofs.tex
\textbf{Proof of Lemma \ref{cam}.}
Consider $O_i$ and let $\mathbb A = \{ A^1, A^2, \ldots, A^n \}$, where $A^k$ is an allocation of $O_i$ that allocates $O_{i,j}$ to agent $j' = ((j+k-1) \mbox{ mod } n)+1.$
Since in $\mathbb A$ each item is allocated to each agent once, $\sum_j \util_i({A^j}) = \sum_j \valu_{j,i}(\items).$ Thus, the worst allocation in set $\mathbb A$ has a utility of at most $ \sum_j{\valu_{j,i}(\items)}/n = \overline{\valu}_i(\items)$ for agent $i$. As a result, $\EMMS_i = \util_i({\worst_i(O_i)}) \leq \overline{\valu}_i(\items)$. 
\qed
\\[0.3cm]
\textbf{Proof of Lemma \ref{cmp}.}
	Consider an instance with one item $\ite$ and $n$ agents, and the following two scenarios for the externalities: 
	\begin{enumerate}[label=(\roman*)]
		\item $\forall_{i \neq j}, \valu_{j,i}(\{\ite\}) = 0$
		\item $\forall_{i \neq j}, \valu_{j,i}(\{\ite\}) = \valu_{i,i}(\{\ite\})$
	\end{enumerate}
	In scenario $(i)$, $\EMMS_1 = 0$, while the extended-proportional share of agent $1$ is $\valu_{1,1}(\{\ite\})/n $. In the second scenario, $\EMMS_1 = \valu_{1,1}(\{\ite\})$, but the extended-proportional share of agent $1$ is still $ \valu_{1,1}(\{\ite\})/n$.
\qed
\\[0.3cm]
\textbf{Proof of Lemma \ref{cac}.}
We know $\util_2(\best_2(O_1)) \geq \EMMS_2$. Furthermore, since $\worst_1(O_1)$ determines the value of $\EMMS_1$, we have: $\util_1(\best_2(O_1)) \geq \EMMS_1.$
\qed
\\[0.3cm]
\textbf{Proof of Lemma \ref{fnp}.}
	Consider a complete bipartite graph $G(X, Y)$ where $X$ represents the bundles of $P$, and $Y$ represents the agents and there is an edge with weight $\valu_{j,i}(P_k)$ between every pair $x_k \in X$ and $y_j \in Y$. Finding $\worst_i(P)$ is equivalent to finding the maximum cardinality matching with minimum weight in $G$. Classic network flow algorithms solve this problem in polynomial time \cite{cormen2009introduction}. 
\qed
\\[0.3cm]
\textbf{Proof of Claim \ref{nohuge}.}
	Consider $L_{i,1}$ (the most valuable bundle of $L_i$ for agent $i$). Trivially, we have $\valu_i(L_{i,1}) \geq \AVG_i(\items)$, and since there is no huge item, $L_{i,1}$ contains at least two items. On the other hand, according to $\LPT$, the items within a bundle arrive in non-increasing order. Therefore, the last item added to $L_{i,1}$ has a value of at most $\valu_i(L_{i,1})/2$ and the total value of $L_{i,1}$ just before the last item arrives must have been at least $\valu_i(L_{i,1})/2$. Furthermore, whenever an item is added to a bundle, that bundle has the minimum value among all the bundles. Therefore, $\valu_i(L_{i,n}) \geq \valu_i(L_{i,1})/2 \geq  \AVG_i(\items)/2$. 
\qed
\\[0.3cm]
\textbf{Proof of Claim \ref{subsetswitch}.}
	Since $P$ is not nice, there exists an item $\ite$ in bundle $P_j$, such that $\valu_i(P_j) > \valu_i(\{b\}) > \valu_i(P_n)$. We modify $P$ as follows: we remove $P_j$ and $P_n$ from $P$ and add two new bundles $A = \{ b \}$ and $B = P_j \cup P_{n} \setminus \{ b\}$ to $P$. Let  $l$ and $l'$ be the indices of the newly added bundles in $P$ (note that  the bundles are rearranged by their decreasing values for agent $i$), such that $j \leq l \leq l' \leq n$ (see Figure \ref{proofassist}). We have $$\valu_i(P_j) > \max(\valu_i(A), \valu_i(B)) \geq \min(\valu_i(A), \valu_i(B)) > \valu_i(P_{n}).$$ By this modification, $\util_i(\worst_i(P))$ increases by a value of at least
	\begin{align*} 
	x_{i,l} \cdot (\max(\valu_i(A), \valu_i(B)) - \valu_i(P_j)) 
	+ \hspace{0.1cm}  x_{i,l'} \cdot  (\min(\valu_i(A), \valu_i(B)) - \valu_i(P_n)),
	\end{align*}
	which is non-negative since $$\valu_i(P_j) + \valu_i(P_n) = \max(\valu_i(A), \valu_i(B)) + \min(\valu_i(A), \valu_i(B)),$$ and $x_{i,l} \leq x_{i,l'}$.
	Let $\lon(P) = \{ P_j \mid \valu_i(P_j) = \valu_i(P_n) \}. $ 
	After each modification, either $V_i(P_n)$ increases, or $|\lon(P)|$ decreases. Therefore, sequence $(\valu_i(P_n), \valu_i(P_{n-1}), \ldots, \valu_i(P_1))$ increases lexicographically by each move, and hence we \textcolor{magenta}{eventually} end up with a nice partition $P'$ after a finite number of modifications.
	

	\begin{figure}[!tbp]
		\centering
		\begin{minipage}[b]{0.49\textwidth}
			\centerline{
				\includegraphics[scale=0.3]{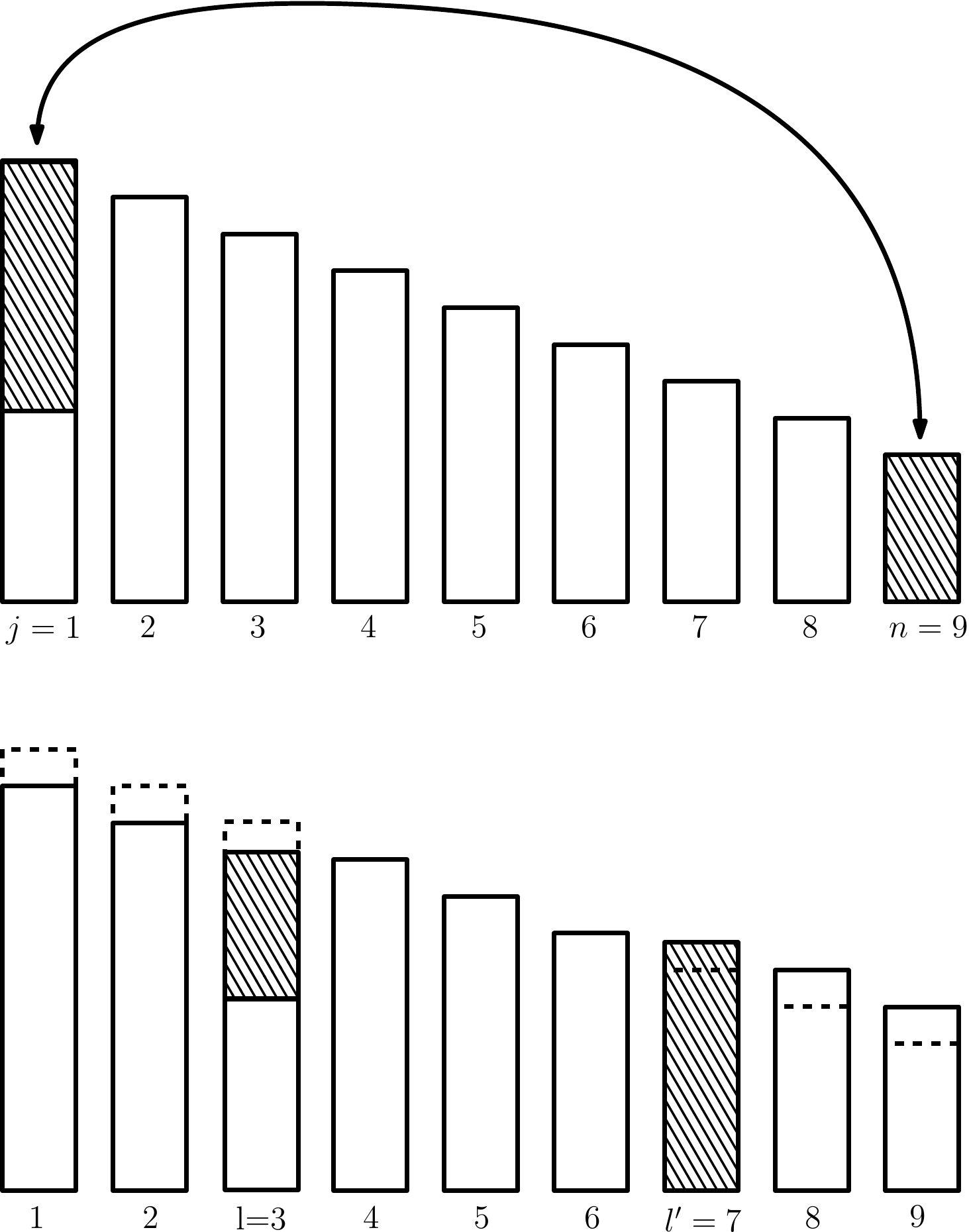}
			}
			\caption{Switching the subsets}
			\label{proofassist}
		\end{minipage}
		\begin{minipage}[b]{0.49\textwidth}
			\centerline{
				\includegraphics[scale=0.9]{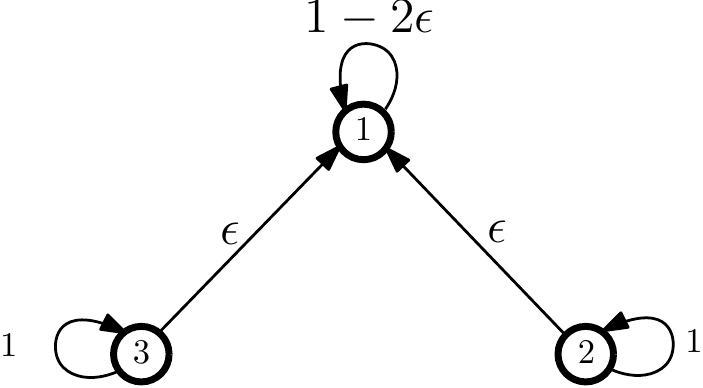}
			}
			\caption{The gap between $\MMS_i$ and $\EMMS_i$ may be large, even with very small externalities}
			\label{c-ex}
		\end{minipage}
	\end{figure}  
	
\qed
\\[0.3cm]
\textbf{Proof of Observation \ref{gap}.}
	Simply consider the influence graph depicted in Figure \ref{c-ex} and three items $\ite_1,\ite_2$ and $\ite_3$ such that
	$
	\valu_1(\{\ite_1\}) = 1 $ and $
	\valu_1(\{\ite_2\}) = \valu_1(\{\ite_3\}) = c/\epsilon
	$,
	where $\epsilon$ is a small constant less than $1/2$.
	For this instance, $\EMMS_1 = (1-2\epsilon) + 2c$, and \textcolor{magenta}{$\MMS_1 = 1-2\epsilon$} 
	which means $\MMS_1 / \EMMS_1 < 1/c $.
\qed
\\[0.3cm]
\textbf{Proof of Observation \ref{minbun}.}
	For every remaining agent $i$, we find a bundle $\Gamma_i$ with the minimum size which meets the expectation of agent $i$ as follows: sort the remaining items in their decreasing values for agent $i$, and add the items to $\Gamma_i$ one by one until the bundle meets the expectation of agent $i$. \textcolor{magenta}{Finally}, it only suffices to select the smallest bundle among these bundles.
\qed
\\[0.3cm]
\textbf{Proof of Lemma \ref{tweak}.}
	Let $\delta$ be a subset of $N_{j,F}$ with the minimum size that satisfies condition (i). Such a set trivially exists. Since no subset of  $N_{j,F}$ is compatible, we have
	$
	\sum_{k \in \delta} \valu_j(B_k) > v_{j,\ell_j}.
	$
	By minimality of $\delta$, no proper subset of $\delta$ satisfies condition (i). It is easy to observe that this can only happen when $\delta$ contains only one agent, say $k$, with $\valu_i(B_k) >  v_{j,\ell_j}$. We show in Lemma \ref{shit} that $|B_k| = 1$; but for now suppose that $\ite$ is the only item in $B_k$.
	
	
	Since $O_j$ is nice\footnote{Recall the niceness from Definition \ref{nicedef}}, there is an index $\ell < \ell_j$, such that bundle $O_{j,\ell} = \{\ite\}$. We modify $M_j$ as follows: we map agent $k$ to $O_{j,\ell}$ and map the former agents of $N_{j,\ell}$ to $F_j$. Clearly, conditions (i) and (ii) preserve for $M_j$ after this process. Again, if no subset of $N_{j,F}$ is compatible, we repeat this modification. Each time we modify $M_j$, the number of indices $\ell$ for which $O_{j,\ell}$ is mapped to an agent $j$ with $O_{j,\ell} = B_j$ increases by one. Therefore, the process terminates after a finite number of modifications.
\qed
\\[0.3cm]
\color{magenta}
Algorithm \ref{alg2} illustrates an overview of the update procedure, which completes the $\textsf{BC}$ algorithm. 
\color{black}
\\[0.3cm]
\begin{algorithm}[H]
	$Resolve = 0$\\
	\While{$Resolve == 0$}{
	$\delta \gets$ Minimum sized subset of $N_{j,F}$, s.t. $\sum_{k \in \delta} V_j(B_k) \geq 1/2 \cdot v_{j,\ell_j}$ \\
	\eIf {$\sum_{k \in \delta} V_j(B_k)  \leq v_{j, \ell_j}$}{
		$N_{j,\ell_j} = \delta$\\
		$N_{j, F} = N_{j,f} \setminus \delta$   \\	
		$\ell_j \gets \ell_j + 1$  \\
		$Resolve+=1$ \Comment{Resolved}\\
	}
	{
		Let $\ell$ be an index s.t. $O_{j,\ell} = B_k$, where $\delta = \{k\}$. \\
		Swap $\delta$ (which is a subset of $N_{j,F}$) with $N_{j,\ell}$. \Comment{One step closer to resolve} \\
	}
	}
	\caption{Update $M_j$}
	\label{alg2}
\end{algorithm}

\begin{lemma}
	\label{shit}
	Suppose that the total value of the items allocated to the agents in $N_{j,F}$ is at least $v_{j,\ell_j}/2$, but $N_{j,F}$ admits no compatible subset, and let $\delta$ be the minimal subset of $N_{j,F}$ that satisfies condition (i). Then, $\delta$ contains only one agent, say agent $k$ and $|B_k|=1$.
\end{lemma}

\begin{proof}
	As mentioned in Lemma \ref{tweak}, it is easy to observe that $|\delta|=1$. Here, we argue that if agent $k$ is the only agent in $\delta$, then $|B_k|=1$.
	As a contradiction, let $z_1$ be the first step of the algorithm that $\delta = \{k\}$, but $|B_k|>1$. In addition, let $z_2$ be the step that $B_k$ was allocated to agent $k$ and let $\ell'_j$ be the expectation level of agent $j$ in step $z_2$. Trivially, we have $z_2 \leq z_1$.

	\begin{claim}
		\label{clm}
		Either $v_{j,n} \geq v_{j,\ell'_j}/2$ or we have $|O_{j,\ell}| = 1$ for all $\ell \leq \ell'_j$.
	\end{claim}
	
	\textbf{Proof of Claim \ref{clm}.}
	If for some $\ell \leq \ell'_j$, $O_{j,\ell}$ contains more than one item, $O_{j,\ell}$ has a proper subset $s$ such that $\valu_j(s) \leq v_{j,\ell}/2$. By the same reasoning as  Claim \ref{subsetswitch}, moving $s$ to bundle $O_{j, n}$ yields a new partition which is at least as good as $O_j$ (See Figure \ref{proofassist}).  Hence, we can assume w.l.o.g. that $v_{j,n} \geq v_{j,\ell'_j}/2$ holds.
	\qed
	
	\vspace{0.3cm}
	Regarding Claim \ref{clm}, we consider two cases.
	
	\vspace{0.3cm}
	\textbf{First}, assume that $|O_{j,\ell}| = 1$ for all $\ell \leq \ell'_j$.
	For this case, at least one of the items in $\bigcup_{\ell \leq \ell'_j}{O_{j,\ell}}$ is not allocated to any agent before step $z_2$, and this item singly meets the expectation of  agent $j$. This contradicts the fact that at step $z_2$, $B_k$ was the minimal set (Note that we supposed $|B_k|>1$).
	
	\vspace{0.3cm}
	\textbf{Second}, assume that $v_{j,\ell_j} \geq v_{j,\ell'_j}/2$. In step $z_2$, the expectation value of agent $j$ equals $v_{j,\ell'}/2$. Furthermore, $\valu_j(B_k) >v_{j,\ell_j}$ which means $\valu_j(B_k) > v_{j,\ell'_j}/2$. On the other hand, $\valu_j(B_k) < v_{j,\ell'_j}$,  otherwise a proper subset of $B_k$ would meet the expectation of agent $j$ in step $z_2$.
	Therefore, in step $z_2$, $\delta = \{ k \}$ is the only compatible set for updating $M_j$ and hence, agent $k$ is mapped to $O_{j,\ell'_j}$. This also implies that $z_2  \neq z_1$, since we supposed that no compatible subset exists in step $z_1$.

	Furthermore, notice that since $|B_k|>1$, no item could singly meet the expectation of any agent, including agent $j$ in step $z_2$. This means that every remaining item in step $z_2$ has the value less than $v_{j,\ell_j}/2$. On the other hand, in all the modifications before step $z_1$, the bundle allocated to the agent in $\delta$ consists of only one item ($z_1$ is te first step that the size of the bundle allocated to the agent in $\delta$ is more than $1$). This means that after step $z_2$, no modification affects the agents that are mapped to bundles $O_{j,\ell}$ for $\ell \leq \ell'_j$. But this contradicts the fact that agent $k$ is mapped to $F_j$ in step $z_1$, because agent $k$ is mapped to $O_{j,\ell'_j}$ and no modification changes $M_j(k)$.
\end{proof}

\textbf{Proof of Lemma \ref{satisfaction}.}
	We want to show that in every step of the algorithm, for each remaining agent $i$, Inequality \ref{eq1} holds. 
	\begin{equation}
	\label{eq1}
	\sum_{j \in \mathcal{S}} w_{j,i} \cdot \valu_i(B_j) \geq \sum_{j < \ell_i} x_{i,j} \cdot v_{i,j} /2.
	\end{equation}

	To prove this, we apply a sequence of exchanges between the bundles allocated to the agents in $\bigcup_{j < \ell_i} N_{i,j}$ and show that in every exchange, value of the expression on the left-hand side of Inequality \eqref{eq1} does not increase
	\footnote{Note that these exchanges are only to prove this lemma, and not in the algorithm.}. 
	Next, we show that after these exchanges, Inequality \eqref{eq1} holds, which means that the Inequality was held for the original allocation. 
	
	Let agent $j$ be the agent in $N_{i,1}$ with the least influence on agent $i$ (i.e., minimizes $w_{j,i}$).  First, we allocate the bundles that belong to the other agents in $N_{i,1}$ to agent $j$ and remove all the agents but agent $j$ from $N_{i,1}$.
	Since agent $j$ has the minimum weight among the agent in $N_{i,1}$, this operation does not increase the left-hand side of Inequality \eqref{eq1}. 
	
	In addition, let agent $j'$ be the agent with $w_{j',i} = x_{i,{1}}$. Since agent $j'$ has the minimum weight among all the agents, $w_{j',i} \leq w_{j,i}$.  Now, let $B_j$ and $B_{j'}$ be the current bundles of agents $j$ and $j'$. Note that if $j'$ is not satisfied yet, $B_{j'} = \emptyset$. If $\valu_i(B_{j'}) < \valu_i(B_j)$, we swap the bundles of $j$ and $j'$. This operation also does not increase the left-hand side of Inequality \eqref{eq1} since we have $w_{j',i} \leq w_{j,i}$. Finally, we exchange the set that agents $j$ and $j'$ belong to: we remove agent $j$ from $N_{i,1}$, and add agent $j'$ to $N_{i,1}$. In addition, if agent $j'$ previously belonged to $N_{i,r}$ for some $r$, we add agent $j$ to $N_{i,r}$. This exchange has no effect on the value of $\sum_{j \in S} w_{j,i} \cdot \valu_i(B_j)$. Furthermore, one can easily observe that after the exchange, condition (ii) holds. 
	
	We repeat the same procedure for $N_{i,2},N_{i,3},\ldots, N_{i,{\ell_i-1}}$.
	After this sequence of exchanges, each $N_{i,j}$ contains one agent $j'$, where $w_{j',i} = x_{i,j}$.  Furthermore, after the exchanges, the second condition for the validity of $M_i$ holds and hence, the value of the items of agent $j'$ for agent $i$ is at least $v_{i,j}/2$.  Therefore, total amount of externalities of the satisfied agents is at least $\sum_{j<\ell_i} x_{i,j}v_{i,j}/2$.
\qed
\\[0.3cm]

%% file: Paper.bbl
\begin{thebibliography}{10}

\bibitem{amanatidis2015approximation}
Georgios Amanatidis, Evangelos Markakis, Afshin Nikzad, and Amin Saberi.
\newblock Approximation algorithms for computing maximin share allocations.
\newblock In {\em International Colloquium on Automata, Languages, and
  Programming}, pages 39--51. Springer, 2015.

\bibitem{anari2010equilibrium}
Nima Anari, Shayan Ehsani, Mohammad Ghodsi, Nima Haghpanah, Nicole Immorlica,
  Hamid Mahini, and Vahab~S Mirrokni.
\newblock Equilibrium pricing with positive externalities.
\newblock In {\em International Workshop on Internet and Network Economics},
  pages 424--431. Springer, 2010.

\bibitem{aziz2016discrete}
Haris Aziz and Simon Mackenzie.
\newblock A discrete and bounded envy-free cake cutting protocol for any number
  of agents.
\newblock In {\em Foundations of Computer Science (FOCS), 2016 IEEE 57th Annual
  Symposium on}, pages 416--427. IEEE, 2016.

\bibitem{bei2012optimal}
Xiaohui Bei, Ning Chen, Xia Hua, Biaoshuai Tao, and Endong Yang.
\newblock Optimal proportional cake cutting with connected pieces.
\newblock In {\em AAAI}, volume~12, pages 1263--1269, 2012.

\bibitem{brams1995envy}
Steven~J Brams and Alan~D Taylor.
\newblock An envy-free cake division protocol.
\newblock {\em The American Mathematical Monthly}, 102(1):9--18, 1995.

\bibitem{branzei2013matchings}
Simina Br{\^a}nzei, Tomasz Michalak, Talal Rahwan, Kate Larson, and Nicholas~R
  Jennings.
\newblock Matchings with externalities and attitudes.
\newblock In {\em Proceedings of the 2013 international conference on
  Autonomous agents and multi-agent systems}, pages 295--302. International
  Foundation for Autonomous Agents and Multiagent Systems, 2013.

\bibitem{bronzei2013externalities}
Simina Br{\^a}nzei, Ariel~D Procaccia, and Jie Zhang.
\newblock Externalities in cake cutting.
\newblock AAAI, 2013.

\bibitem{budish2011combinatorial}
Eric Budish.
\newblock The combinatorial assignment problem: Approximate competitive
  equilibrium from equal incomes.
\newblock {\em Journal of Political Economy}, 119(6):1061--1103, 2011.

\bibitem{dehghani2018envy}
Sina Dehghani, Alireza Farhadi, MohammadTaghi HajiAghayi, and Hadi Yami.
\newblock Envy-free chore division for an arbitrary number of agents.
\newblock In {\em Proceedings of the Twenty-Ninth Annual ACM-SIAM Symposium on
  Discrete Algorithms}, pages 2564--2583. SIAM, 2018.

\bibitem{deuermeyer1982scheduling}
Bryan~L Deuermeyer, Donald~K Friesen, and Michael~A Langston.
\newblock Scheduling to maximize the minimum processor finish time in a
  multiprocessor system.
\newblock {\em SIAM Journal on Algebraic Discrete Methods}, 3(2):190--196,
  1982.

\bibitem{farhadi2017fair}
Alireza Farhadi, MohammadTaghi Hajiaghayi, Mohammad Ghodsi, Sebastien Lahaie,
  David Pennock, Masoud Seddighin, Saeed Seddighin, and Hadi Yami.
\newblock Fair allocation of indivisible goods to asymmetric agents.
\newblock In {\em Proceedings of the 16th Conference on Autonomous Agents and
  MultiAgent Systems}, pages 1535--1537. International Foundation for
  Autonomous Agents and Multiagent Systems, 2017.

\bibitem{ghodsi2017fair}
Mohammad Ghodsi, MohammadTaghi HajiAghayi, Masoud Seddighin, Saeed Seddighin,
  and Hadi Yami.
\newblock Fair allocation of indivisible goods: Improvement and generalization.
\newblock {\em arXiv preprint arXiv:1704.00222}, 2017.

\bibitem{gourves2017approximate}
Laurent Gourv{\`e}s and J{\'e}r{\^o}me Monnot.
\newblock Approximate maximin share allocations in matroids.
\newblock In {\em International Conference on Algorithms and Complexity}, pages
  310--321. Springer, 2017.

\bibitem{graham1969bounds}
Ronald~L. Graham.
\newblock Bounds on multiprocessing timing anomalies.
\newblock {\em SIAM journal on Applied Mathematics}, 17(2):416--429, 1969.

\bibitem{haghpanah2011optimal}
Nima Haghpanah, Nicole Immorlica, Vahab Mirrokni, and Kamesh Munagala.
\newblock Optimal auctions with positive network externalities.
\newblock In {\em Proceedings of the 12th ACM conference on Electronic
  commerce}, pages 11--20. ACM, 2011.

\bibitem{kempe2008cascade}
David Kempe and Mohammad Mahdian.
\newblock A cascade model for externalities in sponsored search.
\newblock {\em Internet and Network Economics}, pages 585--596, 2008.

\bibitem{kurokawa2015leximin}
David Kurokawa, Ariel~D Procaccia, and Nisarg Shah.
\newblock Leximin allocations in the real world.
\newblock In {\em Proceedings of the Sixteenth ACM Conference on Economics and
  Computation}, pages 345--362. ACM, 2015.

\bibitem{kurokawa2016can}
David Kurokawa, Ariel~D Procaccia, and Junxing Wang.
\newblock When can the maximin share guarantee be guaranteed?
\newblock In {\em AAAI}, volume~16, pages 523--529, 2016.

\bibitem{leme2012sequential}
Renato~Paes Leme, Vasilis Syrgkanis, and {\'E}va Tardos.
\newblock Sequential auctions and externalities.
\newblock In {\em Proceedings of the twenty-third annual ACM-SIAM symposium on
  Discrete Algorithms}, pages 869--886. SIAM, 2012.

\bibitem{li2015truthful}
Minming Li, Jialin Zhang, and Qiang Zhang.
\newblock Truthful cake cutting mechanisms with externalities: do not make them
  care for others too much!
\newblock In {\em IJCAI}, pages 589--595, 2015.

\bibitem{mirrokni2012fixed}
Vahab~S Mirrokni, Sebastien Roch, and Mukund Sundararajan.
\newblock On fixed-price marketing for goods with positive network
  externalities.
\newblock In {\em WINE}, pages 532--538. Springer, 2012.

\bibitem{procaccia2014fair}
Ariel~D Procaccia and Junxing Wang.
\newblock Fair enough: Guaranteeing approximate maximin shares.
\newblock In {\em Proceedings of the fifteenth ACM conference on Economics and
  computation}, pages 675--692. ACM, 2014.

\bibitem{robertson1998cake}
Jack Robertson and William Webb.
\newblock {\em Cake-cutting algorithms: Be fair if you can}.
\newblock AK Peters/CRC Press, 1998.

\bibitem{steinhaus1948problem}
Hugo Steinhaus.
\newblock The problem of fair division.
\newblock {\em Econometrica}, 16(1), 1948.

\bibitem{suksompong2017approximate}
Warut Suksompong.
\newblock Approximate maximin shares for groups of agents.
\newblock {\em arXiv preprint arXiv:1706.09869}, 2017.

\bibitem{velez2008fairness}
Rodrigo Velez.
\newblock Fairness and externalities.
\newblock {\em Unpublished Manuscript, University of Rochester}, 2008.

\bibitem{woeginger1997polynomial}
Gerhard~J Woeginger.
\newblock A polynomial-time approximation scheme for maximizing the minimum
  machine completion time.
\newblock {\em Operations Research Letters}, 20(4):149--154, 1997.

\end{thebibliography}
